\documentclass[sigconf,balance=false]{acmart}
\pdfoutput=1
\usepackage{popets}

\setcopyright{popets}
\copyrightyear{2024}

\acmYear{2024}
\acmVolume{YYYY}
\acmNumber{X}
\acmDOI{XXXXXXX.XXXXXXX}
\acmISBN{}
\acmConference{Proceedings on Privacy Enhancing Technologies}

\usepackage{xcolor}
\usepackage{amsmath}
\usepackage{algorithm,algpseudocode}

\usepackage{verbatim} %
\usepackage{subcaption} %
\usepackage{epsfig} %
\graphicspath{{./figures/}}

\usepackage{url}

\usepackage{soul}

\usepackage{pifont}%
\newcommand{\cmark}{\ding{51}}%
\newcommand{\xmark}{\ding{55}}%

\usepackage{hyperref} %
\usepackage[capitalize,nameinlink]{cleveref} %

\newenvironment{tightitemize}%
{
  \begin{list}{$\bullet$}{%
      \setlength{\leftmargin}{10pt}
    \setlength{\itemsep}{0pt}%
      \setlength{\parsep}{0pt}%
      \setlength{\topsep}{0pt}%
      \setlength{\parskip}{0pt}%
      }%
    }%
    {
  \end{list}
  }
  \newcounter{tecounter}
  {
    \begin{list}{\arabic{tecounter}.}{%
        \usecounter{tecounter}
        \setlength{\leftmargin}{10pt}
      \setlength{\itemsep}{0pt}%
        \setlength{\parsep}{0pt}%
        \setlength{\topsep}{0pt}%
        \setlength{\parskip}{0pt}%
        }%
      }%
      {
    \end{list}
    }%

    {
      \begin{list}{}{%
          \setlength{\leftmargin}{10pt}
        \setlength{\itemsep}{0pt}%
          \setlength{\parsep}{0pt}%
          \setlength{\topsep}{5pt}%
          \setlength{\parskip}{0pt}%
          }%
        }%
        {
      \end{list}
    }%

    \DeclareMathOperator*{\argmin}{arg\,min}

    \crefformat{section}{\S#2#1#3} %
    \crefformat{subsection}{\S#2#1#3}
    \crefformat{subsubsection}{\S#2#1#3}
    \crefname{equation}{}{}
    \Crefname{equation}{Equation}{Equations}
    \crefname{alg}{}{}
    \Crefname{alg}{Algorithm}{Algorithms}

\begin{document}

    \title{\bf{Snail}: Secure \bf{S}i\bf{n}g\bf{l}e \bf{I}teration \bf{L}ocalization}
    \newcommand{\naivealgo}{ Localization}
    \newcommand{\algo}{Single Iteration Localization}
    \newcommand{\shortalgo}{SIL}
    \newcommand{\snail}{Turbo the Snail}
    \newcommand{\snailshort}{Turbo}

    \author{James Choncholas}
    \affiliation{
      \institution{Georgia Institute of Technology}
      \city{Atlanta}
      \country{USA}
      }
      \email{jchoncholas3@gatech.edu}

\author{Pujith Kachana}
\affiliation{
	\institution{Georgia Institute of Technology}
	\city{Atlanta}
	\country{USA}
}
\email{pkachana3@gatech.edu}

\author{Andr\'e Mateus}
\affiliation{
	\institution{Ericsson Research}
	\city{Kista}
	\country{Sweden}
}
\email{andre.mateus@ericsson.com}

\author{Gregoire Phillips}
\affiliation{
	\institution{Ericsson Research}
	\city{Sunnyvale}
	\country{USA}
}
\email{greg.phillips@ericsson.com}

\author{Ada Gavrilovska}
\affiliation{
	\institution{Georgia Institute of Technology}
	\city{Atlanta}
	\country{USA}
}
\email{ada@cc.gatech.edu}

\begin{abstract}
	Localization is a computer vision task by which the position and orientation of
a camera is determined from an image and environmental map.
We propose a method for performing localization in a privacy preserving manner
supporting two scenarios: first, when the image and map are held by a client
who wishes to offload localization to untrusted third parties, and second, when
the image and map are held separately by untrusting parties.
Privacy preserving localization is necessary when the image and map are
confidential, and offloading conserves on-device power and frees resources for
other tasks.
To accomplish this we integrate existing localization methods and secure
multi-party computation (MPC), specifically garbled circuits, yielding
proof-based security guarantees in contrast to existing obfuscation-based
approaches which recent related work has shown vulnerable.
We present two approaches to localization, a baseline data-oblivious adaptation
of localization suitable for garbled circuits and our novel \algo{}.
Our technique improves overall performance while maintaining confidentiality of
the input image, map, and output pose at the expense of increased communication
rounds but reduced computation and communication required per round.
\algo{} is over two orders of magnitude faster than a straightforward
application of garbled circuits to localization enabling real-world usage in
\snail{}, the first robot to offload localization without revealing input
images, environmental map, position, or orientation to offload servers.

\end{abstract}


\maketitle

\section{Introduction}
\label{sec:intro}

Visual localization algorithms allow devices to infer their position in a
three-dimensional map from features derived from images.
Popular applications include autonomous vehicles, virtual reality, and robotics
where a device with a camera is moving through an environment and localizing
repeatedly on camera frames to infer its position and orientation over
time~\cite{lu_l3-net_2019, di_rgb-d_2016, auat_cheein_slam_2010}.

Incorporating privacy into visual localization is an active area of research.
Consider a lightweight robot equipped with a camera.
Privacy preserving localization allows the robot to offload the localization
task to more powerful resources without revealing its location, its map, or
images taken by its camera.
Offloading conserves onboard power and frees resources for other tasks.
Privacy is important in this setting because the robot may inadvertently take
pictures of sensitive information, or people who do not want their picture
taken.
Ethical or legal requirements may prevent the images from being shared outside
the device domain.
Privacy preserving localization also enables settings when the map and camera
image are held by different parties, or when both are held by the same party
who is trying to prevent data exfiltration by requiring an attacker to
compromise multiple systems simultaneously.
While technologies have emerged to make privacy preserving localization more
efficient, none make formal security statements and many of the claimed
security properties have been broken~\cite{speciale2019, dusmanu2021,
  brachmann_visual_2021, liu_range-based_2019}.

Previous attempts at secure localization rely on techniques based on
obfuscation, like line cloud transformations~\cite{speciale2019} and
adversarial affine subspace embeddings~\cite{dusmanu2021}.
Such obfuscation-based techniques do not make formal security statements, but
they have low overhead and run efficiently.
Consequently, these techniques have been shown to be vulnerable to attacks
which can recover the input image~\cite{chelani2021privacy,
  pittaluga2023ldpfeat}.
Recently, differential privacy has been used to address one such attack however
privacy is inversely related to the amount of data processed
sequentially~\cite{pittaluga2023ldpfeat}.
While this is appropriate for some applications, it is not appropriate for
repeated invocation, common when localizing on sequential camera frames.

As recent innovation demonstrates, applying privacy-preserving techniques to
visual localization that prioritize efficiency over formal security guarantees
leaves weaknesses in the protection these techniques afford.
At the same time, general purpose secure computation can present performance
obstacles with high computation overhead and communication costs.
Localization is not well suited to execution under homomorphic encryption due
to the algorithm's high multiplicative depth, heavy use of division, and
requirement for floating point data representation.
Garbled circuits are a better fit for localization, however they are
communication intensive.
What is needed is a co-design approach that meets privacy expectations with
practical efficiency.

In this work, we investigate privacy preserving visual localization using
secure multi-party computation (MPC).
We first develop a data-oblivious implementation in which we run a standard
visual localization algorithm under MPC based on non-linear optimization.
This provides the first known implementation of visual localization with formal
security guarantees drawn from general purpose secure computation.
While secure, this na\"ive approach has inefficiencies which stem from the
iterative nature of localization algorithms.
The popular localization approach we consider (minimizing reprojection error
via non-linear least squares) is composed of two nested iterative steps which
are unfriendly to data-oblivious execution, one being optimization via gradient
descent and the other being singular value decomposition (SVD).
To address the performance issues these iterative steps present we develop a
novel \algo{} approach, leveraging the fact that localization is typically run
repeatedly on sequential camera frames.
We modify the outer iterative algorithm, gradient descent, to run each
iteration independently in a way which maintains security for a series of
localization runs.
Then, we address the inner iterative algorithm, the SVD, by finding the optimal
number of iterations a~priori, which in practice does not depend on secret
input data.
The resulting {\em secure, non-linear and iterative localization}~--~{\bf
SNaIL}~--~runs two orders of magnitude faster than a na\"ive adaptation of
localization to MPC.

In summary, this work makes the following contributions:
\begin{tightitemize}
  \item We design a novel \algo{} (\shortalgo{}) method for visual localization
  which increases round complexity in exchange for orders of magnitude runtime
  improvement and better privacy properties.
  \item We present a simulation-based definition of security for
  privacy preserving localization\footnote{Source code available at
    \url{https://github.com/secret-snail/localization-server}}.
  \item We experimentally evaluate \shortalgo{} against a
  data-oblivious baseline using two different visual localization algorithms
  and two different MPC frameworks, EMP~\cite{emp-toolkit} and
  ABY~\cite{demmler2015aby}.
  \item We demonstrate real-world practicality with \snail{}, the first robot
  to offload localization without revealing the view of its camera or its
  position and orientation in the environment.
\end{tightitemize}

\section{Preliminaries}
\label{sec:background}
This work builds on the goals of previous works seeking to integrate
privacy-enhancing features into localization.
Our specific focus is the following two privacy goals: 1) preventing image and
map reconstruction and 2) maintaining confidentiality of the pose.
These two goals represent the most illustrative standards with which to
contrast our methods to the existing state-of-the-art~\cite{speciale2019,
  chelani2021privacy, dusmanu2021}.

\subsection{Secure Computation}
\label{sec:sec_comp}
Secure computation describes a cryptographic field that seeks to allow multiple
parties to compute a function over secret inputs.
Secure multiparty computation (MPC) protocols serve a similar purpose as
trusted execution environments (TEEs) but do not require hardware support and
are not vulnerable to a class of side channel attacks that affect
TEEs~\cite{van2017telling, lee2017hacking, lee2017inferring}.
When applied in practice, MPC is often collaborative in the sense in that each
party has secret inputs to a function, for example a bid in an auction, and the
function output, e.g. the winning bid, is learned by all participants.
MPC, however, may also be used to offload computation from a weak device to
stronger resources such that the weak device with the secret inputs plays a
minor role in the protocol.

Often MPC protocols are constructed with either Shamir or additive secret
sharing protocols~\cite{shamir1979share}, else garbled circuit-based
protocols~\cite{yao1982protocols,yao1986generate} briefly introduced next.
Secret share-based approaches rely on splitting data into pieces in a way that
retains homomorphic properties such that each participant operates on a
ciphertext indistinguishable from randomness.
Garbled circuit-based approaches instead structure computation as a boolean
circuit where each wire value, zero or one, is represented by an encryption
key.
The circuit is evaluated gate by gate using input wire keys to decrypt output
wires which are themselves input keys to the next level of gates.
In the two party case, one party plays the role of the generator who creates
the garbled circuit and the other party plays the role of the evaluator,
decrypting the wires received from the generator.
The encryption keys for the first set of wires in the circuit corresponding to
plaintext inputs owned by the evaluator are sent using an Oblivious Transfer
protocol~\cite{rabin2005exchange}.
This allows the evaluator to learn the appropriate wire label and ensures that
generator does not learn which label the evaluator is requesting.

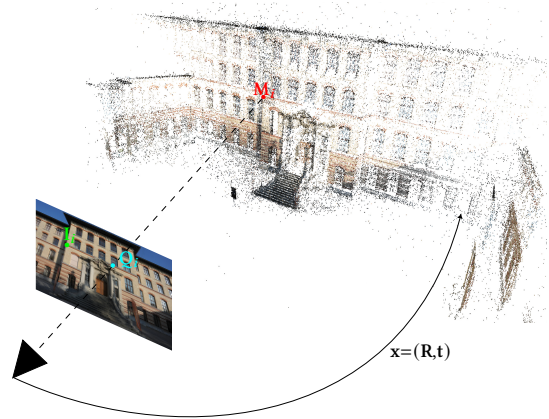
\begin{figure}
  \vspace{1cm} \centering \def\svgwidth{.4\textwidth}
  \begingroup%
  \makeatletter%
  \providecommand\color[2][]{%
    \errmessage{(Inkscape) Color is used for the text in Inkscape, but the package 'color.sty' is not loaded}%
    \renewcommand\color[2][]{}%
  }%
  \providecommand\transparent[1]{%
    \errmessage{(Inkscape) Transparency is used (non-zero) for the text in Inkscape, but the package 'transparent.sty' is not loaded}%
    \renewcommand\transparent[1]{}%
  }%
  \providecommand\rotatebox[2]{#2}%
  \newcommand*\fsize{\dimexpr\f@size pt\relax}%
  \newcommand*\lineheight[1]{\fontsize{\fsize}{#1\fsize}\selectfont}%
  \ifx\svgwidth\undefined%
    \setlength{\unitlength}{1125bp}%
    \ifx\svgscale\undefined%
      \relax%
    \else%
      \setlength{\unitlength}{\unitlength * \real{\svgscale}}%
    \fi%
  \else%
    \setlength{\unitlength}{\svgwidth}%
  \fi%
  \global\let\svgwidth\undefined%
  \global\let\svgscale\undefined%
  \makeatother%
  \begin{picture}(1,0.66666667)%
    \lineheight{1}%
    \setlength\tabcolsep{0pt}%
    \put(0,0){\includegraphics[width=\unitlength]{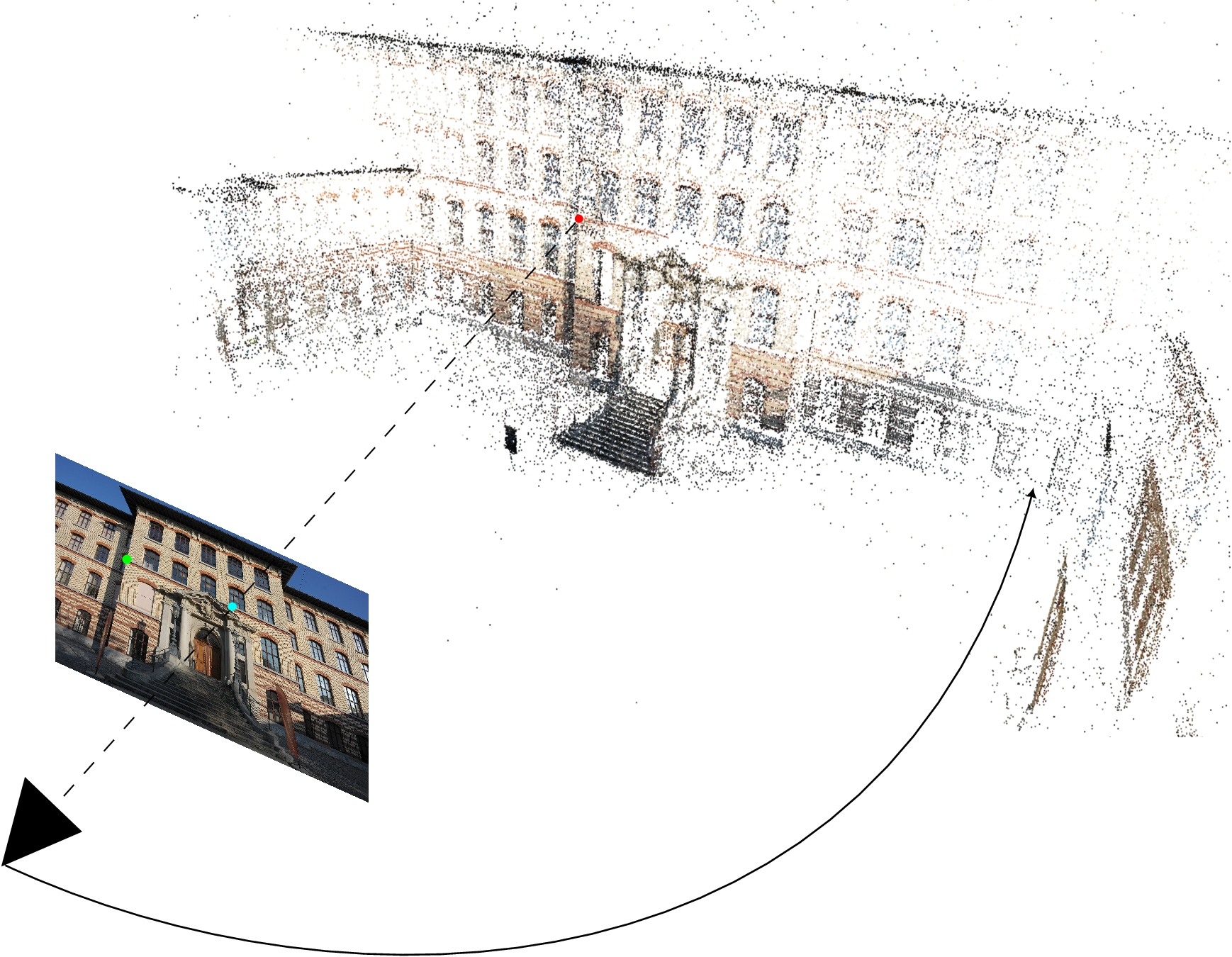}}%
    \put(0.7070424,0.11340818){\color[rgb]{0,0,0}\makebox(0,0)[lt]{\lineheight{1.25}\smash{\begin{tabular}[t]{l}${\scriptstyle \mathbf{x} = (\mathbf{R},\mathbf{t})}$\end{tabular}}}}%
    \put(0.4508839,0.60358958){\color[rgb]{1,0,0}\makebox(0,0)[lt]{\lineheight{1.25}\smash{\begin{tabular}[t]{l}${\scriptstyle \mathbf{M}_i}$\end{tabular}}}}%
    \put(0.09742143,0.33224895){\color[rgb]{0,1,0}\makebox(0,0)[lt]{\lineheight{1.25}\smash{\begin{tabular}[t]{l}${\scriptstyle \mathbf{I}_i}$\end{tabular}}}}%
    \put(0.20032488,0.28645005){\color[rgb]{0,1,1}\makebox(0,0)[lt]{\lineheight{1.25}\smash{\begin{tabular}[t]{l}${\scriptstyle \mathbf{Q}_i}$\end{tabular}}}}%
  \end{picture}%
\endgroup%
 \caption{Depiction of the PnP problem.
    The goal is to find the pose $\mathbf{x}$, which minimizes the error
    $\mathbf{dI} = \|\mathbf{Q}_i - \mathbf{I}_i\|$ between the image measured
    points $\mathbf{I}_i$ and the map points $\mathbf{M}_i$ projected to the image
    $\mathbf{Q}_i$.
  }
  \label{fig:pnp}
\end{figure}

\subsection{Visual Localization}
Visual localization comprises a series of approaches that infer a device's
location from local features extracted from visual data, traditionally derived
from a set of 2D or 3D input images depicted in \Cref{fig:pnp}.
This process commonly employs structure-based approaches, representing scenes
through point clouds and feature matching between retrieved 2D images and 3D
point clouds.
Point clouds, i.e. maps, may be generated using Structure-from-Motion
(SfM)~\cite{yang_high-precision_2009}, constructed simultaneously, or stored
from previous mapping of the environment.
Such approaches that rely on sharing image features from local 2D images (such
as sending them to an external server for analysis) risk leaking private
information about the environment to the processing entity, as features are
known to be susceptible to image
inversion~\cite{weinzaepfel_reconstructing_2011}.
In practice, several studies demonstrate the recovery of image content from
gradient-based feature descriptors, and/or locations, in algorithms such as
SIFT~\cite{lowe2004} and HOG~\cite{balntas2017hpatches}, rendering such methods
vulnerable to privacy threats when revealed, as one can reasonably reconstruct
the image.

Alternatives to such structural approaches exist in the realm of learned
localization~\cite{liu_range-based_2019}, which replaces part or all of the
localization pipeline with machine-learning based optimization.
Several emerging families of machine-learning based localization methods
demonstrate high degrees of performance comparable or better than image
retrieval techniques~\cite{dusmanu2019d2}.
However, these learning-based methods currently fail to scale in complexity
beyond small and relatively simple
scenes~\cite{brachmann_visual_2021,cavallari2019let}, and to be as accurate as
geometric methods~\cite{sattler2019}.
Furthermore, such methods are comparably vulnerable to image recovery from
model features~\cite{pittaluga2019revealing}.

Classical techniques i.e., non-learning based, to solve the Perspective-n-Point
(PnP) problem consist of three main approaches.
The first focuses on robust estimation, by removing outliers.
A minimal solver is used in conjunction with a RANSAC~\cite{fischler1981} loop.
In the context of camera pose estimation, the minimal number of 2D-3D
correspondences required to obtain a solution is three~\cite{persson2018}.
The second type of approaches solves for the pose using two-steps 1) estimate
the 3D points in the camera coordinate system and; 2) solve the 3D-3D pose
problem, which has a closed-form solution~\cite{umeyama1991}.
An example of this type of approach is presented in \cite{lepetit2009}.
The third type of approach focuses on solving an optimization problem to obtain
the pose.
Variants of this approach exploit specific problem formulations as in
\cite{lu2000}, or in specific solvers, such as Levenberg-Marquardt (LM) and
Gauss Newton (GN).

In the field of computer vision, the PnP problem~\cite{lepetit2009} consists of
estimating the pose (rotation and translation) of a camera with respect to a
world coordinate system.
The pose is estimated by exploiting a set of 2D feature locations (e.g. points
in an image) and their corresponding 3D points (e.g.~a 3D map of the
environment).
If a camera can move freely in 3D space its pose has six degrees of freedom,
three rotational and three translational, about the three Cartesian axis.
Applications in robotics and AR/VR often use PnP solvers as part of larger
processing pipelines including control~\cite{chaumette2006} and
mapping~\cite{schonberger2016}.
This work focuses only on pose estimation with the PnP problem.

\subsection{Privacy-Preserving Localization}
\label{sec:prelim_pploc}
Several existing methods implement privacy preserving visual localization and
other image query technologies.
Speciale et al.~\cite{speciale2019} introduce a line cloud-based method for
localization that relies on transforming 3D point clouds to line clouds to
prevent the types of image inversion techniques described
in~\cite{weinzaepfel_reconstructing_2011}.
This work has been extended/adapted to applications in Simultaneous
Localization and Mapping (SLAM)~\cite{geppert2021privacy} and
SfM~\cite{geppert2020privacy} systems.
However, the authors note several limitations to the privacy preservation
qualities of this method themselves, namely that reconstructing the secret
image used as localization input becomes easier with repeated invocation.
Chelani et al.~\cite{chelani2021privacy} further demonstrate additional threats
posed by this method that undermine its privacy guarantees in that the secret
image can be fully reconstructed in just one invocation.
A comparison to this work is shown in~\Cref{tab:related_comparison}.

Outside of methods directly implementing privacy-preserving localization, other
work applies privacy-enhancing technologies to related image query technologies
in ways that could be extended to visual localization algorithms.
Dusmanu et al.~\cite{dusmanu2021} introduce a privacy-preserving method for
extracting image features using adversarial affine subspace embeddings.
Recently, this approach was shown to be vulnerable to attacks that can recover
the input image~\cite{pittaluga2023ldpfeat} precluding its use for
privacy-preserving localization.
The solution proposed to address attacks presented
in~\cite{pittaluga2023ldpfeat} relies on differential privacy which is not
appropriate for the localization setting we consider.
Since privacy loss increases with the number of queries, localization on an
unchanging sequence of images (like a stationary camera) may quickly exceed a
reasonable privacy budget.

Engelsma et al.~\cite{engelsma2022hers} demonstrate a homomorphically encrypted
representation search that could be further adapted to existing visual
localization algorithms.
This represents a promising frontier, but is currently in the early stages of
application and does not scale to such processes in its current form.

The performance envelope of privacy preserving localization proposed in this
work is not well suited for time-critical applications like AR/VR, as we later
show.
There are, however, better suited applications which are less time-sensitive
such as Google Visual Positioning System, Facebook Livemaps, and Microsoft
Azure Spatial
Anchors~\cite{googlevps,facebooklivemaps,microsoftspatialanchors}.
These services offer maps supporting localization for thier users.
Privacy preserving localization would allow a user to localize an image or
video using a mapping service without the user sharing their image or location
for privacy reasons, and without the map provider needing to share their map
for commercial reasons.

As another example of where privacy is important in localization, consider a
delivery robot which must navigate private facilities to deliver packages.
The robot is not trusted to learn the full map of these facilities, nor is it
allowed to share images it has taken as they may reveal people's faces or
confidential information from the facility.
Privacy preserving localization keeps images on device and the full map of the
facilities off device.
The robot may run MPC-based localization with two map servers where one is
owned by the robot's manufacturer and the other by the facility who uses their
confidential map as input to the protocol.
The robot does not learn the full map of the facility, nor do the map servers
learn the robot's location or images.

\subsection{Homomorphic Encryption and Localization}
The localization algorithms considered are particularly challenging to adapt to
execution via homomorphic encryption (HE).
Most problematic, is that localization is a high depth computation with
significant data dependencies between steps.
For example, to perform one iteration of Levenberg-Marquardt localization, over
1,000 ciphertext-ciphertext divisions and over 7,000 multiplications are
required in a high depth circuit, as we later show in~\Cref{fig:arith_ops}.
The high degree of dependencies between input data does not suit batching based
optimizations for HE schemes popular in BGV, BFV and CKKS where multiple
plaintexts are encrypted and operated on together~\cite{bgv,bfv,ckks}.
This, and the large fractional precision required, suggests localization would
best suited for a boolean-based FHE scheme.
A popular implementation of one such scheme is tfhe-rs which claims to evaluate
a boolean gate in 8.5~ms~\cite{TFHE-rs}.
Given there are 4979 AND and XOR gates per floating point multiplication, for
example in EMP-Toolkit's floating point circuits~\cite{emp-toolkit}, the time
required to perform one iteration of Levenberg-Marquardt based localization is
estimated to be at least 423 seconds for the multiplications alone.
For these reasons, homomorphic encryption is not further considered in favor of
an MPC-based approach.
While MPC has higher communication complexity than HE, i.e.~communication is
proportional to the size of the computation vs the size of just the inputs, we
later show this cost is practical.

\begin{table}[]
  \centering
  \begin{tabular}{rccccc}
    \multicolumn{1}{c}{}
    &
    \begin{tabular}
      [c]{@{}c@{}}
      Private \\ Features
    \end{tabular}
    &
    \begin{tabular}
      [c]{@{}c@{}}Private \\ Map
    \end{tabular}
    &
    \begin{tabular}
      [c]{@{}c@{}}Private \\ Pose
    \end{tabular}
    &
    \begin{tabular}
      [c]{@{}c@{}}Round \\ Complexity
    \end{tabular}
    &
    \begin{tabular}
      [c]{@{}c@{}}Speed
    \end{tabular}
    \\ \cline{2-6}
    \multicolumn{1}{r|}{\cite{speciale2019}} & \xmark{} \cite{chelani2021privacy} & \xmark{} & \xmark{} & 1 & fastest \\
    \multicolumn{1}{r|}{\cite{dusmanu2021}} & \xmark{} \cite{pittaluga2023ldpfeat} & \xmark{} & \xmark{} & 1 & fastest \\
    \multicolumn{1}{r|}{DO} & \cmark{} & \cmark{} & \cmark{} & 1 & slow\\
    \multicolumn{1}{r|}{\shortalgo{}} & \cmark{} & \cmark{} & \cmark{} & data dependent & fast
  \end{tabular}
  \vspace{2ex}
  \caption{Comparison to related privacy preserving localization techniques.
    DO is our data-oblivious garbled circuits baseline and \shortalgo{} is our
    proposed \algo{} technique.
  }
  \label{tab:related_comparison}
\end{table}

\section{Security Model}
\label{sec:secmodel}
The security guarantees we aim to capture encompass two settings, as previously
mentioned.
In the first setting, we envision a lightweight client who would like to
leverage the resources of two or more capable offload servers.
The client has all inputs to the computation, the camera image and 3D map, and
does not trust the offload servers with its secret input data.
Offloading computation from the client saves power and frees resources for
other tasks.
In the second setting, the data is instead split between two parties where the
input images are known to the client and the map is owned by a third party.
The map owner does not want to share the map, but wants to help devices
localize using it.
For example, a drone using visual localization may be captured thus is not
trusted to hold a confidential map, while the map owner should not learn the
drone's location or camera images.
We consider the semi-honest security model, where participants are not trusted
with secrets but are expected to follow the protocol prescribed of them.

While semi-honest security is often considered a stepping stone to more robust
security models, even achieving semi-honest security is difficult especially in
a performance sensitive setting.
The authors of~\cite{speciale2019} note their line cloud transformations
technique is inversely proportional to the amount of data processed in series.
It was later discovered that the secret image can be fully reconstructed after
just one localization run~\cite{chelani2021privacy}.
This was significant as it showed simply knowledge of which algorithm is used
to extract features from the camera's image along with the obfuscated feature
coordinates is enough to reasonably recover the image.
Due to this, we find it necessary to introduce the more rigorous,
simulation-based definition of security for privacy preserving localization.
Complexity-based cryptography, or provable security, was first introduced in
the context of encryption in 1984~\cite{shafi1984probabilistic}.
Schemes considered secure under this definition are sequentially composable:
the schemes can be composed without requiring proving that the overall scheme
is secure.
Given computer vision tasks are often composed e.g.\ running localization
repeatedly on a stream of camera frames or in combination with random sample
consensus algorithms, we choose this framework in which to define security for
privacy preserving localization.

In MPC, security against semi-honest adversaries is formulated as a game
referred to as the real-ideal paradigm.
Informally, there is an attacker who runs a protocol and remembers all the
messages they receive and send.
If another algorithm called a simulator can generate a convincing set of
messages that resembles what a real attacker would see during protocol
execution without actually running the protocol, the scheme is considered
secure.
Next we define this game more formally in the context of localization using
notation and format from Lindell~\cite[Chapter-6,
  Section-4.2]{lindell2017simulate}.
The inputs to localization, namely the set of features taken from an image and
from a map, as well as the output pose are considered secret, though the secret
map features may be owned by the device or another party.
This definition considers a streaming model where multiple poses are computed
from a set of images and maps.
We note that localization algorithms are deterministic (nonrandomized) by
nature which allows for a slightly simpler definition of security as the output
can always be computed from inputs alone and thus parties in the protocol may
be simulated individually.

\begin{tightitemize}
  \item Let $l(\textbf{I},\textbf{M}, x)$ be a probabilistic
  deterministic polynomial-time localization algorithm which takes as input
  a set of $\textbf{I} = \{\mathbf{i}_0, \ldots, \mathbf{i}_{i-1}\}$ image
  feature locations, a set $\textbf{M} = \{\mathbf{m}_0, \ldots,
    \mathbf{m}_{i-1}\}$ map feature locations, and initial pose estimate
  $\mathbf{x}$.
  $l$ returns a refined pose, $\mathbf{x}_{\text{refined}}$.
  \item Let $L(\textbf{F})$ compute $l$ on every member of a set of
  $j$ image features, map features, and initial pose estimates
  \break
  $\textbf{F}=\{\{\textbf{I},\textbf{M},\textbf{x}\}_0, \ldots,
    \{\textbf{I},\textbf{M},\textbf{x}\}_{j-1} \}$.
  \item Let $\pi$ be an $n$-party protocol which computes $L$.
  The view of the $k$\nobreakdash-th party during an execution of $\pi$ is
  \linebreak $\text{view}_{k}^{\pi}(\textbf{F}) = (w, r^k; m_0^k, \ldots, m_h^k)$
  where $w$ is the party's input, $r^k$ is the $k$-th party's randomness, and
  $m_n^k$ represents the $n$-th message that it received.
  The output of $\pi$ to party $k$ is $\text{output}_{k}$ and is implicitly part
  of the party's view because $\pi$ is deterministic.
\end{tightitemize}
\vspace{1ex}

Protocol $\pi$ securely computes $L$ in the presence of static, semi-honest
adversaries if there exists probabilistic polynomial-time algorithms $S_0,
  \ldots, S_{n-1}$ such that
\begin{equation}
  \{ S_k(1^{\kappa}, w_k,
  \text{output}_{k}) \}_{\textbf{F}, \kappa}
  \mathrel{\stackrel{\makebox[0pt]{\mbox{\normalfont\tiny c}}}{\equiv}}
  \{\text{view}_{k}^{\pi} (\textbf{F})\}_{\textbf{F}, \kappa}
  \label{eq:security_def}
\end{equation}
where, the symbol
$\mathrel{\stackrel{\makebox[0pt]{\mbox{\normalfont\tiny c}}}{\equiv}}$ refers
to computational indistinguishability with security parameter $\kappa$.
Given the two views, the probability they can be distinguished by a non-uniform
polynomial-time algorithm is negligible, i.e. $O(\frac{1}{2^\kappa})$.

This definition has some differences compared to a straight forward application
of simulation-based definition of security from generic MPC.
First, this definition considers a ``streaming'' model where multiple poses are
computed from a set of images and maps $\mathbf{F}$.
The streaming model captures temporal dimension where localization is executed
multiple times sequentially.
This is a common real-world setting for localization where a device is
continuously localizing on a stream of images from a camera, as in robotics,
AR/VR, etc. Directly applying the standard definition of security to
localization would require each localization execution to be simulated
independently.
The streaming model however supports more efficient protocols, as we later show
with our \algo{} approach, because localization executions do not need to be
independent.
From a security perspective, it is sufficient that a sequence of localization
executions do not leak any information about their inputs, but simulating each
independently is unnecessary.

This definition gives insight into where previous privacy preserving
localization attempts have failed.
Line cloud transformations of Speciale et.
al. cannot be simulated; the
lines contained within the $\text{view}$ are not known to the simulator and
thus the difference between the real and simulated views can be easily
distinguished, as practically proven by Chelani et. al.~\cite{speciale2019,
  chelani2021privacy}.
Granted, line clouds were only meant to prevent image reconstruction and do not
attempt to hide map feature locations $\textbf{M}$ or the output pose
$\mathbf{x}$ (see~\Cref{tab:related_comparison}).
However, just considering input feature locations, an adaptation of the above
definition to remove the map $\textbf{M}$ and output pose $\mathbf{x}$ still
suggests line cloud transformations are not secure.
Furthermore, line cloud transformations do not consider the streaming model;
authors note their technique is inversely proportional to the amount of data
processed in series.
A concrete definition of security is needed to capture repeated invocation via
a streaming model to avoid this leakage.

\section{Overview of Privacy Preserving Localization}
\label{sec:overview}
Our primary contribution is a practical method and implementation of privacy
preserving localization $L(\mathbf{F})$ on a set of image and map features
$\textbf{F}$ meeting the security definition in \Cref{sec:secmodel}.
This requires overcoming two types of challenge, the first being practical
challenges stemming from the large size and depth of the computation required.
Secondly, is addressing the iterative nature of localization, for which we
introduce \algo{}.
But first, we introduce localization as a functionality.

\begin{algorithm}
  \caption{Pose Estimation}\label{alg:generic_localization}
  \begin{flushleft}
    \hspace*{\algorithmicindent} \textbf{Inputs:} Image feature locations $\mathbf{I}$, \\
    \hspace*{\algorithmicindent} \hspace*{8ex}Map feature locations $\mathbf{M}$, \\
    \hspace*{\algorithmicindent} \hspace*{8ex}Initial pose estimate $\mathbf{x}$. \\
    \hspace*{\algorithmicindent} \textbf{Output:} Pose $\mathbf{x_{\text{refined}}} = l(\mathbf{I}, \mathbf{M}, \mathbf{x}).$
  \end{flushleft}
\end{algorithm}

This work considers privacy preserving localization in two contexts.
The first is a client-server setting where a lightweight client wishes to
offload localization to more capable servers, e.g.\ when the client is a mobile
device with limited power or computational resources.
In this setting, the client encrypts its input image and map features and sends
them to the offload servers.
The offload servers compute on the ciphertext and each returns a ciphertext
representing the result back to the client.
The client can then construct the plaintext resultant pose including position
and orientation $x = l(\mathbf{I}, \mathbf{M})$.
Both input features and output pose are considered confidential in contrast to
related approaches in which only the input features are
secret~\cite{speciale2019}.
The offload servers are expected to be within different administrative domains,
for example two providers offering private localization as a service.
The details of how data is sent to and retrieved from offload servers depends
on the underlying MPC protocol, covered later, however at a high level,
protocols based on secret shares split input data into additive pieces while
garbled circuit protocols send inputs encoded as circuit wire labels.

The second setting considered is where a client holds the features extracted
from their image $\mathbf{I}$, but the 3D map $\mathbf{M}$ is held by a third
party.
The client and third party interact to compute the pose.
The client should not learn anything about the map other than what can be
inferred from the pose, and the third party should learn (informally) nothing.
In the real world, the full map $\mathbf{M}_{\text{full}}$ likely contains more
features than the client's image $\mathbf{I}$.
In order to keep $\mathbf{M}_{\text{full}}$ confidential, we assume the
existence of an alignment functionality to match the client's features in
$\mathbf{I}$ to the subset of features shared by the map $\mathbf{M}_{full}$,
specifically, $\mathbf{M} = align(\mathbf{I}, \mathbf{M}_{full})$.
We believe the alignment functionality is a reasonably straightforward
application of private information retrieval~\cite{chor1998private} and leave
concrete instantiation to future work.
When alignment is composed with localization, i.e. $x = l(\mathbf{I},
  \text{align}(\mathbf{I}, \mathbf{M}_{\text{full}}))$, the client learns neither
the exact map features $\mathbf{M}$, nor map features which did not have a
corresponding image feature, i.e.
$\{x : x \in \mathbf{M}_{\text{full}} \text{ and } x \notin \mathbf{M}\}$.

In both settings considered, the core challenge to instantiating $l$ in a
privacy preserving way is in addressing the iterative nature of localization.
Iterative algorithms like localization are not friendly to running under MPC
because they are not data-oblivious, meaning control flow depends on secret
data i.e.\ the convergence criteria.
To run an iterative algorithm under MPC, an upper bound of iterations must be
specified a~priori, regardless of convergence at runtime, to ensure the rate of
convergence is not leaked.
If the rate of convergence were leaked, it reveals sensitive information
allowing the devices position to be inferred, especially when localization is
run repeatedly like on sequential camera frames (as is common in robotics,
AR/VR, and autonomous vehicles).
The iteration leakage problem has been studied previously in the context of
privacy preserving SAT solvers~\cite{ppsat, mpclp} and there is no direct
solution; an upper bound of iterations must be executed to hide the real number
of iterations required for convergence.
Localization is particularly problematic as it contains an outer gradient
descent iterative algorithm, each iteration of which computes the singular
value decomposition (SVD) to (pseudo) invert a matrix which is itself an
iterative algorithm.
Our \algo{} approach addresses this, instantiated using an appropriate
localization algorithm and MPC protocol given the context.

We have identified two key observations to address the wasted work required by
the iterative localization algorithm considered.
The first key is when computing the SVD: our analysis shows the number of
iterations required does not in fact depend on the input data in the case of
the considered localization algorithm.
As such, the optimal number of SVD iterations can be fixed a~priori, which
turns out to be data-independent and wastes no work.
The second key is to execute each iteration of the outer gradient decent
algorithm independently.
Doing so is secure under the assumption that a certain number of localization
executions will be run in series, as we later show more formally.
This turns out to be quite amenable to how localization is run in practice.
In the context of autonomous vehicles, robotics, AR/VR, etc.\ localization is
invoked repeatedly on sequential camera frames.
Thus, with these observations and algorithm-cryptography co-design, we then
demonstrate practicality of our approach with \snail{}, a proof of concept
robot which securely offloads localization to navigate its environment.

\section{Design}
\label{sec:design}
The design of a privacy preserving localization method for MPC concerns four
critical choices: selecting an appropriate localization algorithm, an MPC
protocol and implementation, a data representation, and addressing issues with
data-obliviousness given the iterative nature of optimization-based
localization.
The design is presented in this order as each choice is dependent on the
previous.
Design choices are justified them in terms of their wall clock runtime and
security implications, where applicable.

\subsection{Localization Algorithm}
\label{sec:loc_alg}
The vast number of localization algorithms is at odds with the significant
engineering effort required to build them in privacy preserving frameworks.
This process is time-consuming as algorithms must be re-written from scratch to
use privacy preserving data types, be made data-oblivious, and correctly stage
data in and out of the privacy preserving framework.
Thus, this work considers only the most popular approach to localization
consisting of solving a non-linear least-squares problem which minimizes the
reprojection error~\cite{Hartley2004}.
Specifically, find the pose $\mathbf{x}$ (position and orientation) which
minimizes the error $\mathbf{dI}$ between the set of projected image points
$\mathbf{Q} = \text{proj}(\mathbf{x},\mathbf{M})$, and the corresponding
measured image points $\mathbf{I}$.
$\mathbf{M} = [\mathbf{x}^M,\mathbf{y}^M,\mathbf{z}^M]^\top$ is the set of 3D
map points.
Notice the one-to-one correspondence between each point in $\mathbf{M}$ and
$\mathbf{I}$, i.e., for each 3D point in the map, we can observe it in the
image, a measured 2D point.
Formally, the problem is defined as
\begin{equation}
  \argmin_{\mathbf{x}}
  \hspace{1ex} \sum_{i=1}^n \mathbf{dI}^2, \, \text{with} \, \mathbf{dI} =
  \|\mathbf{Q}_i - \mathbf{I}_i\|,
  \label{eq:reproj_error}
\end{equation}
where
$n$ is the number of 2D and 3D points.
Subscript $i$ denotes the $i$th point $\mathbf{Q}_i = [x_i^Q, y_i^Q]$ projected
into the image plane from 3D point $\mathbf{M}_i = [x_i^M, y_i^M, z_i^M]^\top$
and similarly for the $i$'th measured image point $\mathbf{I}_i = [x_i^I,
  y_i^I]$.

Point projection is computed with intrinsic camera parameters (focal length
$\mathbf{f} = [f_x,f_y]$ and image center $\mathbf{c} = [c_x,c_y]$), pose
$\mathbf{x}$ (represented as Euler finite rotation matrix $\mathbf{R}$, and
translation $\mathbf{t} = [t_x,t_y,t_z]^\top$), and 3D world point
$\mathbf{M}_i$ as
\begin{equation}
  proj(\mathbf{x}, \mathbf{M}_i) =
  \begin{bmatrix}
    f_x & 0   & c_x \\
    0   & f_y & c_y \\
    0   & 0   & 1
  \end{bmatrix}
  \begin{bmatrix}
    r_{1,1} & r_{1,2} & r_{1,3} & t_x \\
    r_{1,1} & r_{2,2} & r_{2,3} & t_y \\
    r_{3,1} & r_{3,2} & r_{3,3} & t_z
  \end{bmatrix} \begin{bmatrix}
    x_i^M \\
    y_i^M \\
    z_i^M \\
    1
  \end{bmatrix} ,
  \label{eq:projection}
\end{equation}
where $r_{ij}$
represent the entries of $\mathbf{R}$ at row $i$ and column $j$.
A depiction of the PnP problem in presented in \cref{fig:pnp}.
For a detailed explanation of 3D point projection refer to \cite[Section
  3.3]{ma2004}.

To solve the non-linear least-squares problem, we consider two optimization
algorithms: Gauss-Newton (GN) and Levenberg- Marquardt (LM) with Fletcher's
improvement~\cite{fletcher1971modified}.
The pose update $\mathbf{dx}$ for each is defined as
\begin{equation}
  \begin{split}
    \text{GN: } & \mathbf{dx} = (\mathbf{J}^\top \mathbf{J})^{-1}
    \mathbf{J}^\top \mathbf{dI} \\ \text{LM: } & \mathbf{dx} = \Bigl(
    \mathbf{J}^\top \mathbf{J} + \lambda \ \text{diag}(\mathbf{J}^\top \mathbf{J})
    \Bigl)^{-1} \mathbf{J}^\top \mathbf{dI},
  \end{split}
  \label{eq:optimization}
\end{equation}
where the Jacobian matrix $\mathbf{J}$ is computed numerically
from the right-hand side of \Cref{eq:projection} by perturbing each element of
the pose by epsilon using a one hot encoded vector and projecting the points
using the perturbed pose.
The reprojection error $\mathbf{dI}$ is computed as the difference between the
projected image points $\mathbf{Q}$ and the measured points $\mathbf{I}$.
LM is presented at a high level in \Cref{alg:localization} with a more detailed
description in reference~\cite{nocedal1999}.
Note \Cref{alg:localization} computes the pose for a single image and map
$l(\mathbf{I}, \mathbf{M})$, while this work is concerned with pose estimation
for a sequence of images and maps $L(\mathbf{F})$ as described by the security
model in \Cref{sec:secmodel}.

\begin{algorithm}
  \caption{Levenberg-Marquardt Pose Estimation}\label{alg:localization}
  \begin{flushleft}
    \hspace*{\algorithmicindent} \textbf{Inputs:} Image features $\mathbf{I}$, Map features $\mathbf{M}$, \\
    \hspace*{\algorithmicindent} \hspace*{8ex}Initial pose estimate $\mathbf{x}$. \\
    \hspace*{\algorithmicindent} \textbf{Public Parameters:} Intrinsic camera parameters $\mathbf{f}$ and $\mathbf{c}$. \\
    \hspace*{\algorithmicindent} \hspace*{20ex} Convergence criteria $c$, Fletcher's $\lambda$. \\
    \hspace*{\algorithmicindent} \textbf{Output:} Refined pose $\mathbf{x_{\text{refined}}} = l(\mathbf{I}, \mathbf{M}, \mathbf{x}).$
  \end{flushleft}
  \begin{algorithmic}[1]
  \While{not converged}
    \State{}$\mathbf{Q} \gets proj(\mathbf{x}, \mathbf{M})$
    \Comment{Point projection}

    \For{$\text{each degree of freedom } d \gets 1, 6$}
      \Comment{Jacobian}
      \State{}$\mathbf{J}[:][d] \gets \frac{\partial}{\partial \mathbf{x}_d}  proj(\mathbf{x}+(\text{one\_hot}_d*{\epsilon}),\mathbf{M})$,
    \EndFor{}

    \State{}$\mathbf{dI} \gets \mathbf{Q}_i - \mathbf{I}_i$
    \Comment{Reprojection error}

    \State{}$\mathbf{dx} \gets \Bigl( \mathbf{J}^\top \mathbf{J} + \lambda \ \text{diag}(\mathbf{J}^\top \mathbf{J}) \Bigl)^{-1} \mathbf{J}^\top \mathbf{dI}$
    \State{}$\mathbf{x} \gets \mathbf{x} + \mathbf{dx}$
    \Comment{Pose update}

    \If{$\sum_{i=1}^n \|\mathbf{dI}\|^2 \leq c$}
    \Comment{Convergence criteria}
    \State{}Converged. Output $\mathbf{x}$.
    \EndIf{}
  \EndWhile{}
  \end{algorithmic}
\end{algorithm}

\begin{figure}[t!]
  \centering
  \includegraphics[width=.9\columnwidth]{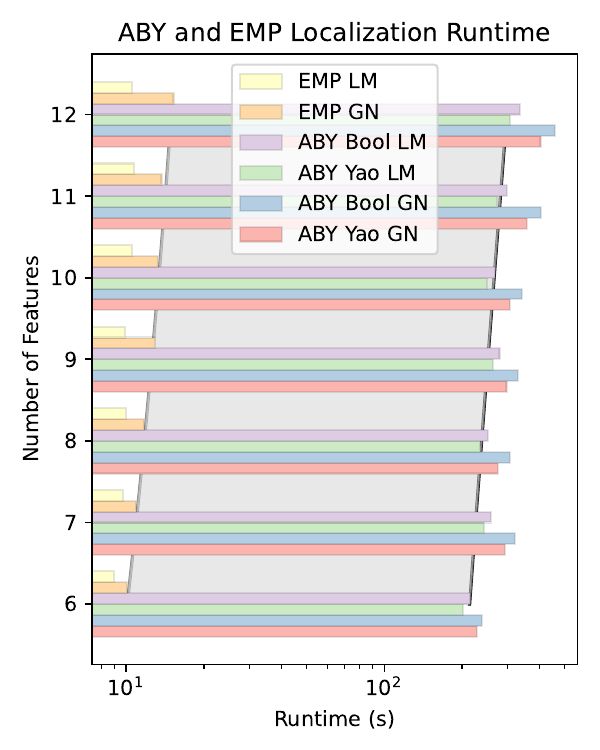}
  \caption{Time to localize using ABY and EMP secure computation frameworks on
    feature data from ETH3D~\cite{schops2017multi}.
    Each measurement is an average of three trials using randomly selected points.
    The same random sequence of points are used across measurements.
    Note the log scale of the x-axis.
    The highlighted area is the difference between the slowest EMP configuration
    and the fastest of ABY.
  }
  \label{fig:aby_vs_emp}
\end{figure}

LM is selected for its wide popularity, real-world usage, and generality as it
is the default PnP algorithm in OpenCV~\cite{opencv_library} and accepts any
number of input point correspondences making it useful in combination with
other localization related tasks like RANSAC~\cite{fischler1981}.
GN is considered as it requires the same set of linear algebra operations as
LM, thus is easy to implement, and also requires fewer matrix multiplications
at the cost of inverting a larger matrix.
While there are many more algorithms, LM and GN are the most popular and
require a narrow variety of operations, i.e., matrix
arithmetic~\cref{eq:reproj_error}, point projection~\cref{eq:projection}, Euler
finite rotation transformation~\cref{eq:projection}, and a matrix (pseudo)
inverse via SVD~\cref{eq:optimization}.
Furthermore, these algorithms can be directly applied, or easily adjusted, to
other computer vision problems like relative pose
estimation~\cite{bartoli2004}, homography \& image
alignment~\cite{Hartley2004}, and Structure-from-Motion
(SfM)~\cite{schonberger2016}.

\subsection{MPC Library}
This work compares localization algorithms implemented using two MPC libraries
secure in the semi-honest setting, ABY and EMP.
ABY is a flexible MPC library which performs secure computation using
arithmetic, boolean, and Yao's garbled circuits~\cite{demmler2015aby}.
ABY is flexible in that it supports switching between protocols to allow the
most efficient protocol to be used for differnt parts of a computation.
EMP, on the other hand, is a framework with many protocols, one of which is a
highly efficient garbled circuits implementation of the semi-honestly secure
half-gates protocol and optimizations~\cite{zahur2015two,
  kolesnikov2008improved}.
One major difference between these implementations that is not immediately
clear is ABY pre-generates circuits before executing them while EMP generates
and executes garbled circuits on the fly.
Circuit pre-generation has the disadvantage in that large circuits require
large amounts of memory.
Running our data-oblivious LM-based localization with 6 input correspondences
and default security parameter results in out of memory errors consuming over
100~GB of RAM (including swap) in ABY.
EMP on the other hand completes using less than 1 GB of RAM.

While ABY and EMP's garbled circuits have similar theoretical performance,  in
practice the lightweight nature of EMP gives it a significant advantage for
this application, as shown in~\Cref{fig:aby_vs_emp}.
ABY's circuit pre-generation means that circuit structure can be re-used after
being built in contrast with EMP's on-the-fly style execution.
This is an artifact of the MPC framework and not the underlying protocols.
However, only considering the online time of ABY and ignoring all setup costs
like base oblivious transfer and circuit generation, ABY is still two orders
magnitude slower than the total runtime of EMP due to the impact of the
library's heavy utilization of the memory allocator.

This finding is surprising because in theory the two frameworks rely on similar
garbled circuits protocols yet runtime varies widely.
Simply measuring the number of gates in a localization circuit and multiplying
by the time it takes to evaluate a single logic gate is not an accurate way to
estimate the runtime of frameworks which pre-generate circuits.
The impact of ABY's memory allocations reduces the performance of the framework
regardless of protocol, be it boolean or Yao.
Note ABY's arithmetic protocol is not of interest in this context because
localization requires operating on floating point data types, as later
discussed.
We suspect there are optimizations we could make to our ABY implementation like
aggregating operations across multiple integers to amortize certain
cryptographic overheads and switching between Arithmetic, Boolean, and Yao
protocols dynamically, however, this optimization process is even more
time-consuming than the initial implementation within the privacy preserving
framework and it is unlikely to make up the orders of magnitude difference to
EMP.
As such, further analysis only considers EMP.

\subsection{Data Representation}
Given promising results from using alternative ways to represent data in the
field of machine learning, we consider various representations in privacy
preserving localization~\cite{gupta2015deep, lesser2011effects}.
Plaintext arithmetic operations on floating point data are accelerated on
modern hardware making localization on floating point data fast (around 10~ms
for 10k operations) and thus other representations like fixed point are not
usually considered.
In contrast, when performing floating point operations under the MPC frameworks
considered, there are no floating point functional units.
Garbled circuit logic gates are evaluated in the same manner no matter how data
is represented.
For this reason we consider both fixed and floating point data types of various
width for privacy preserving localization.

Pose estimation is typically computed with double precision floating point
therefore the first question is whether localization is possible using single
precision floats and fixed point.
The second question is then of performance.
For measurement purposes we use the ETH3D dataset~\cite{schops2017multi}, a 3D
map and series of images taken from known locations commonly used to measure
accuracy and performance of PnP solvers.
Due to integer over/underflow, localization with 32-bit fixed point precision
is not possible.
On the other hand, 64-bit fixed point localization does converge but only when
the SVD, a sub-step in both GN and LM, is computed with floating point.
While even larger width fixed point representations may allow the SVD to
converge successfully, performance measurements suggest even 64-bit fixed point
is not useful.

\begin{figure}[t!]
  \centering

  \includegraphics[width=.9\columnwidth]{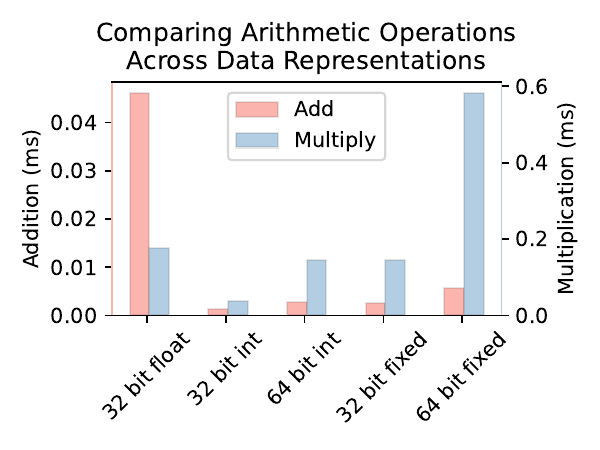}
  \caption{Time to perform addition and multiplication using
    floating point and fixed point data representation of various width
    with EMP's semi-honestly secure half-gates protocol.
  }
  \label{fig:float_vs_fixed}
\end{figure}

\begin{figure}[t!]
  \centering
  \includegraphics[width=.9\columnwidth]{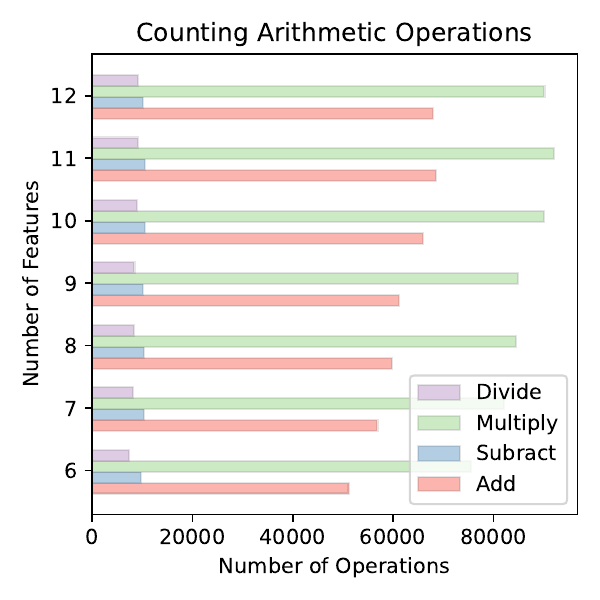}
  \caption{Arithmetic operations performed during LM localization using
    \algo{}.
    GN exhibits similar behavior.
  }
  \label{fig:arith_ops}
\end{figure}

Even though fixed point arithmetic is simpler than floating, comparing 64-bit
fixed point (the smallest width fixed point representation for which
localization converges) to 32-bit floating is not a clear performance win due
to the increase in data width.
In~\Cref{fig:float_vs_fixed} we see 64-bit fixed point addition is faster than
32-bit floating point but multiplication is slower.
This is expected as addition has linear complexity in the number of bits while
multiplication has quadratic.
Thus, determining which offers better performance for localization requires
knowledge of how many of each operation are performed.
\Cref{fig:arith_ops} shows multiplication is the dominant operation
in both LM and GN localization, thus
it is fastest to compute with floating point representation.
From these results, we project operations using 64-bit fixed point to be over
three times slower than 32-bit floating.
We also notice 64-bit fixed point converges $90\%$ less frequently than 32-bit
floating point on the ETH3D dataset, suggesting even wider data types may be
necessary to achieve similar convergence properties to 32-bit floating point.
Due to this performance degradation, we do not further analyze the impact of
fixed point on localization accuracy or the rate of convergence.
We leave the exploration of alternative floating point representations to
future work e.g.~half precision or bfloat16~\cite{bfloat}, as they have more
general implications to plaintext localization.

\subsection{Garbled Circuits for Computation Offload}
Garbled circuit protocols are often considered to run between two parties, each
of which has secret inputs to the function being computed.
In the offload setting we consider, instead the offload servers have no inputs
to the function and are not allowed to learn the output.
Achieving this requires some minor logistical modifications to the garbled
circuits protocol setup.

Recall two party garbled circuit protocols have two roles, a generator and an
evaluator.
The generator creates garbled circuit, which in practice mostly consists of
repeatedly evaluating the AES blockcipher.
The garbled circuit, i.e.~a set of ciphertexts, is then sent to the evaluator
who decrypts parts of the received ciphertext, a process that again mostly
consists of evaluating AES.
The evaluator learns a subset of the ciphertext representing inputs called
input wire labels via a primitive named oblivious transfer.
The evaluator then uses the input wire labels to decrypt the garbled circuit.

In the setting where all plaintext inputs and outputs are owned by the client,
for whom we additionally want to minimize computation and network usage, the
simplest way to get the correct input wire labels from generator to evaluator
is to have the generator send all labels to the client, who then forwards the
appropriate label per wire onwards to the evaluator.
We make a slight optimization from the above approach to reduce the bandwidth
requirements of the client.
The client chooses a random seed and sends it to the generator who uses a
pseudorandom function and the seed to create the wire labels for the garbled
circuit.
Note that the EMP toolkit framework already uses a seeded PRF to generate the
garbled circuit for performance reasons (/dev/random is slow); we have instead
allowed the client to choose this seed which is not an issue since the client
is the only party who supplies the secret input.
After, the generator creates the garbled circuit, they send it to the evaluator
excluding the input wire labels.
The client, knowing the circuit seed, can generate the input wire labels
without interaction and send them to the evaluator.
The circuit is evaluated as usual and the output labels are sent back to the
client who can decode their plaintext meaning.
This saves the client from needing to receive all wire labels from the
generator.
Instead, the client may generate them directly and send to the evaluator.
Since the same seed is used to generate input wire labels and the rest of the
garbled circuit, correctness holds.
Note in the case of EMP and the half-gates protocol, the client must also know
the global circuit delta~\cite{kolesnikov2008improved} as part of the
generation process.
Intuitively, it is okay for the client to learn the seed and delta as the
client is the only party who supplies secret inputs and is allowed to learn the
output.
While this is appropriate for the offload setting, in the setting where the
input image and map are known to different parties, this optimization is not
appropriate.

\begin{figure*}[t]
  \centering
  \includegraphics[width=1.9\columnwidth]{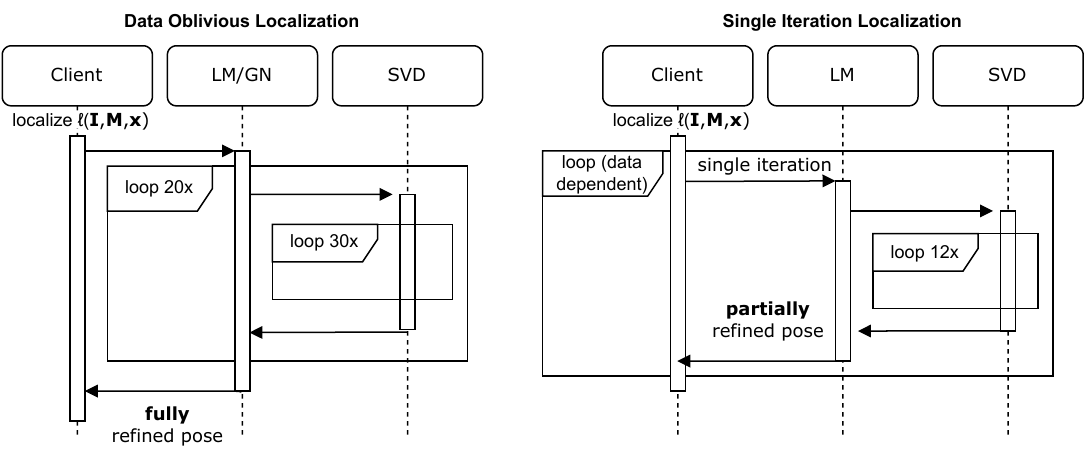}
  \caption{UML diagrams of a na\"ive data-oblivious adaptation of PnP
    localization via gradient descent and \algo{}.
    On the left, a constant upper bound of optimization iterations (default of 20
    LM iterations in OpenCV) are required.
    On the right, the number of optimization iterations is data-dependent which in
    practice is much fewer than 20 but comes at the cost of additional round
    complexity.
    Regarding the inner SVD algorithm, the number of iterations can be reduced from
    a constant upper bound on the left (default of 30 in Eigen/LAPACK) to the
    optimal number (12) on the right using public knowledge about the input
    distribution, namely two QR sweeps per singular value and one singular value
    for each physical degree of freedom (2*6).
  }
  \label{fig:umlalgo}
\end{figure*}

\section{\algo{}}
\label{sec:algo}
The localization algorithms considered (i.e. \Cref{alg:localization}) are
iterative by nature; they run a sequence of gradient descent steps until
convergence is reached.
Furthermore, each iteration requires computing the (pseudo) inverse of a matrix
which is itself an iterative process when using the popular algorithm from
Demmel and Kahan~\cite{demmel1990accurate}, used in the Eigen and LAPACK
libraries to reduce the bidiagonal form to singular values.
Executing such iterative algorithms in a data-oblivious setting (like under
MPC) have a critical drawback: a fixed number of iterations must be executed
regardless of the rate of convergence.
Either work is wasted by running more than necessary iterations, or
completeness of the solution is degraded by running too few.

Our approach maintains the invariant that control flow does not depend on
secret data but avoids wasting work by running on multiple inputs.
The outer iterative algorithm, LM optimization via gradient decent, is broken
apart such that each step runs independently.
Then, stringing the steps from multiple localizations together hides the
convergence rate of each individual localization.
The inner algorithm, the SVD, is addressed through finding that in practice
there is a constant and optimal of number of iterations which does not depend
on secret data.
We call the combination of the two techniques \algo{}.

Concerning the outer gradient decent algorithm, instead of running the entire
decent under MPC, we run each optimization step individually, hence the name
\algo{}.
Specifically, instead of running LM optimization until the pose meets the
convergence criteria, \shortalgo{} takes a single gradient descent step towards
the optimal pose and returns the intermediate result, as shown
in~\Cref{fig:umlalgo}.
The client can then decide after each step if convergence has been met and can
refine the pose by running \shortalgo{} with the same image and map.
When convergence has been met, they may provide a new image and map and start
to localize with the new inputs.
The MPC parties cannot tell whether the client is refining the pose from the
previous step or if they are computing the first step for a new image, because
of the security guarantees of the MPC protocol.
This approach requires that there is a stream of images to localize on.
If the client were to only localize a single image, the offload servers would
learn how many iterations it takes to compute the pose which is a security
issue as previously discussed.
Since localization is often evaluated repetitively on a sequence of camera
frames, repeated evaluation is desirable in practice.
Later in~\Cref{sec:eval} we better quantify how many images constitutes a
stream and the security properties.

There is still one remaining problem in \shortalgo{} as described.
A single gradient descent step requires inverting a matrix which is itself an
iterative algorithm whose convergence is data-dependent.
Just as before, this suggests an upper bound of iterations is required however
the previous solution to reduce the number of iterations does not work as this
problem is nested within the previous.
It turns out there is an optimal number of iterations which is constant and
known beforehand, and thus is easy to encode in the garbled circuit.
To see why, we describe more about how the SVD of a matrix is computed, the
underlying operation in matrix inversion.
To compute the SVD of a matrix it is first reduced to an upper bidiagonal form
using Householder transformations.
Then, the SVD is computed from the bidiagonal matrix using QR factorization,
eliminating the upper diagonal entries as it sweeps across the diagonal in an
iterative fashion.
The number of iterations taken to reduce the bidiagonal form to singular values
depends on the values of the matrix being decomposed.
A general rule of thumb is two QR sweeps per singular value~\cite[p.
  165,]{parlett1998symmetric}
but a bidiagonal matrix whose diagonal entries are
equal to the superdiagonal takes fewer QR factorization
iterations~\cite[Table 2.]{demmel1990accurate}
The key insight we rely on is that the number of singular values of the matrix
is always six as each corresponds to a physical degree of freedom.
The question is then if localization algorithms require inverting a matrix
which, after householder transformations (the first step in computing the SVD),
result in diagonal entries equal to the superdiagonal.
The geometric implication of this phenomenon with respect to localization is
not immediately clear, however, for every experiment performed including real
world images and those from the ETH3D dataset, there are no cases where a
single diagonal entry is within $\pm 10^{-5}$ of the superdiagonal, much less
all diagonal entries.
If there is a case where more QR sweeps are required, the accuracy of the
localization may be affected, however, this suggests the number of iterations
required to compute the SVD is a constant known ahead of time.
Namely, twelve QR sweeps are required for convergence, two for each singular
value, and there are six singular values each corresponding to a physical
degree of freedom.
The takeaway is that the \emph{fixed upper bound} of iterations required to
compute the SVD can be set as the optimal number of iterations which in
practice does not depend on the secret inputs.

As we later quantify, \algo{} is faster than executing one large garbled
circuit for each localization run when considering a stream of images, however
it comes at the cost of additional round complexity.
Instead of participants interacting once per image to share a garbled circuit,
they must instead do so for every intermediate gradient decent step.
In practice round complexity is not an issue as connectivity is already
assumed.
Furthermore, the network RTT time is much faster than the time to run one
iteration, so the effect of the additional rounds is minimal.

\subsection{Implementation}
\label{sec:impl}
As a baseline, we have implemented a data-oblivious adaptation of the
localization algorithms described in~\Cref{sec:loc_alg}.
The baseline maintains the invariant that control flow does not depend on
secret data, thus an upper bound of iterations is encoded for each of the
iterative routines, namely gradient descent and SVD.
The upper bounds are taken from the plaintext algorithms, e.g.\ the plaintext
LM algorithm from OpenCV runs at most 20 iterations, and the plaintext LAPACK
SVD algorithm runs at most 30 iterations.
The baseline is visualized in \Cref{fig:umlalgo} and has been implemented for
GN and LM localization, in both the EMP and ABY frameworks.

We compare the baseline against \algo{}, also implemented for GN and LM in both
ABY and EMP in C++, all of which are available under the MIT
license\footnote{\url{https://github.com/secret-snail/localization-server}}.
The ABY implementation is 3,887 lines of code, EMP is 1,924 lines, and
reference plaintext implementation used for automated testing is 848 lines.
Each implementation is an independent library that includes all necessary
linear algebra operations depending only on the respective MPC library,
including MPC-based adaptations of the SVD algorithm from Demmel and
Kahan~\cite{demmel1990accurate}.

\section{Evaluation}
\label{sec:eval}
Our evaluation is focused on both security and performance.
Because this work leverages existing localization algorithms, we do not discuss
plaintext-related metrics like localization accuracy, ability to converge, or
sensitivity to noise as the privacy preserving implementations share these
properties with plaintext algorithms which have been studied
extensively~\cite{campos2019}.

\subsection{Security Analysis of \algo{}}
\label{sec:eval_security}
This security analysis considers an MPC-hybrid model meaning we assume the
underlying garbled circuit protocol is secure and prove security of \algo{}
according to the definition in \Cref{sec:secmodel}.
This implies the input feature locations, map, and pose remain
confidential over a single iteration, as it is simply an invocation of garbled
circuits which we assume is secure.
It may seem like assuming the MPC protocol is secure leaves nothing
left to prove, however, this is not the case.
Simply running each iteration of gradient decent under MPC independently
to localize \textit{one} image is not secure as the number of invocations is leaked.
The important question is regarding security in the streaming setting,
i.e.\ over multiple invocations of localization.

As an example, suppose a client has one image which requires 8
localization iterations to converge.
The MPC participants learn how many times they were called and in turn have
learned something about the quality of the initial pose estimate.
In general, more iterations are required if the initial pose estimate is far
away from the true pose; fewer are required if the initial pose estimate is
nearby.
Convergence speed also depends on tunable algorithm parameters and the presence
of degenerate solutions in the path of gradient descent.
It is unclear if such information could be used to directly reconstruct the
input image similar to attacks on line cloud obfuscation in previous
work~\cite{speciale2019, chelani2021privacy}.
What is clear, however, is this leakage has a negative impact on a higher level
notion of privacy over time.
For example, a camera moving through an area which is easier to localize
compared to the rest of its environment requires fewer iterations in one
specific area.
The MPC participants learn when the client is in the easily-localized area and
when it is not.
This leakage is captured by the definition in
\Cref{eq:security_def}, namely, the view of participants does not contain $j$
which is the size of $\mathbf{F}$ (the number of images and maps).
This leakage is critical to address and is primary focus of our security
analysis.

Recall the key insight of \algo{} is to break each iteration of gradient descent
apart.
To estimate the pose of a single image, \shortalgo{} must be invoked repeatedly,
until the pose converges.
If computing the pose of one image, this would still reveal the number of invocations
to the MPC participants, which is a security problem as previously discussed.
But if computing the pose on a stream of images, the MPC participants cannot
tell which pose refinement iterations belong to which images.
This hides how many iterations were required to converge for each image.
If the MPC participants also do not know how many images were used as input,
they cannot infer the number of iterations per image, which addresses the leakage.
Decoupling the number of iterations from the higher level localization executions
is the key to security.

There is one issue with the security argument as described.
We have not quantified how many images count as a ``stream'' or more
specifically, how many times \shortalgo{} must be invoked to hide the
association of gradient decent steps to input images.
For example, a stream containing one single image which requires three gradient
decent iterations to compute the pose doesn't offer very good security properties
because the MPC participants can see only three iterations were performed
and so the initial pose estimate must have been very close to the true pose.
It turns out, it depends on a public parameter of the LM and GN localization
algorithms, namely the upper bound on the number of gradient decent steps.
OpenCV's implementation of LM localization defines a maximum number of
iterations of 20 meaning if the pose hasn't converged after 20 iterations,
localization will halt.
It follows that if at least 20 iterations are performed, the probability of an
adversary outputting the correct number of iterations per image is no better
than guessing.

More generally, if \shortalgo{} is invoked $o$ times and the publicly known
upper bound of iterations is $c$, the number of images which could
have been used as input is between $\frac{o}{c}$ and $o$ (where $o > c$).
Thus, the probability of an adversarial offload server $A$ correctly outputting
the number of input images $i$ from the number of iterations it ran, $o$, is
given by:
\begin{equation}
	\text{Pr[A(o) = i]} < \frac{1}{o - \frac{o}{c}} +
	\frac{1}{p(\kappa)}
\end{equation}
where $p(\kappa)$ is a polynomial in
security parameter $\kappa$, an artifact which appears due to the negligible
advantage an attacker has when distinguishing garbled circuit wire labels.
Thus, we conclude that \algo{} is secure when invoked a minimum of $c$ times.
In practice, the default maximum from OpenCV is $c=20$ and since localization
is often called repeatedly (more than 20 times) on camera frames, this is not a
difficult requirement to meet.

Next, we prove security of \algo{} in the simulation-based definition from
\Cref{eq:security_def} using the notation from \Cref{sec:secmodel}.
The proof follows trivially from the previous discussion.
\begin{lemma}
	\shortalgo{} securely implements localization $\textbf{L}$ over a set of
	image and map feature sets $\mathbf{F}$ in the presence of static semi-honest
	adversaries in the MPC-hybrid model when invoked more than $c$ times.
\end{lemma}

\begin{proof}
	Simulator $S$ constructs a view for the offload server by running
	one iteration of localization $l$ on randomly chosen inputs
	$\textbf{I}$, and $\textbf{M}$, appending messages of the party's
	respective role to the view.
	The messages for each iteration of $l$ may be simulated given we consider the
	MPC-hybrid model.
	The simulator does this a random number times which is at least $c$.
	The sets of messages corresponding each iteration are indistinguishable
	from one another and from those in the real protocol,
	thus the two ensembles are indistinguishable.
\end{proof}

There are two important takeaways.
First, the more often the client calls the localization function, the stronger
the privacy up to a maximum of $c$ iterations which may be fully simulated.
This is the opposite of prior work where privacy guarantees weaken with
subsequent invocations.
Second, the security guarantees of \shortalgo{} are in fact stronger than a
na\"ive data-oblivious adaptation.
The latter reveals exactly how many input images were used for localization.
\shortalgo{} reveals only a probability distribution of how many input images
were used as input $[o,\frac{o}{c}]$.

In summary, \algo{} does not reveal any information about the input features
(both image and map), or the pose when run more than $c$ times in series which
in practice is 20.
This is true for any input including the case when the same input image and map
is repeating (corresponding to a stationary camera) or for inputs
that change over time (a moving camera).
Previous attacks on privacy preserving localization rely on exploiting
knowledge of the feature locations to reconstruct the input image.
Such attacks are not possible on \algo{} as the image features, map features,
and all other partial information remains confidential.

\begin{figure}[t!]
	\centering
	\includegraphics[width=.9\columnwidth]{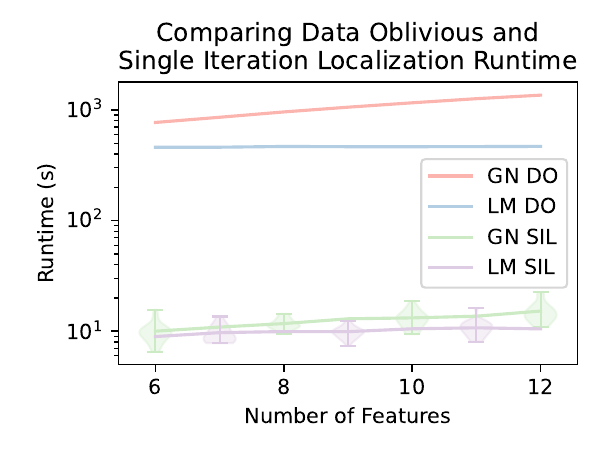}
	\caption{Time to compute localization using data oblivious
		and \algo{} with EMP.
		Data oblivious uses a fixed upper-bound of 20 optimization iterations (the
		default number of LM iterations in OpenCV) and 30 SVD iterations (the default
		in LAPACK) where as \algo{} is data dependent.
	}
	\label{fig:loopleak_vs_dataobl}
\end{figure}

\begin{figure*}[t!]
	\centering
	\begin{subfigure}{\columnwidth}
		\centering
		\includegraphics[width=.9\columnwidth]{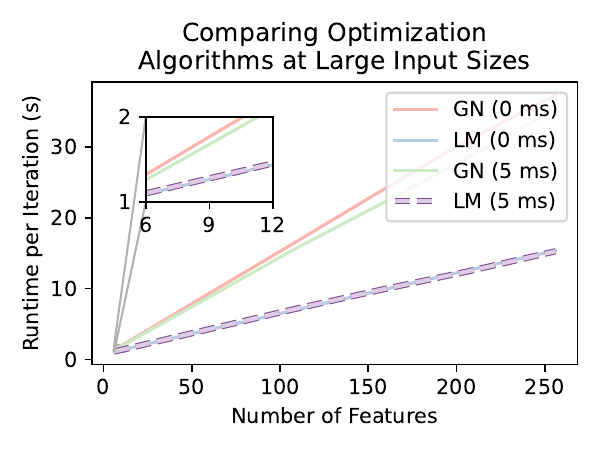}
		\caption{Wall clock runtime of \algo{} on large input size with 0~ms and
			5~ms of latency introduced between offload servers.}
		\label{fig:emp_full}
	\end{subfigure}
	\hfill%
	\begin{subfigure}{\columnwidth}
		\centering
		\includegraphics[width=.9\columnwidth]{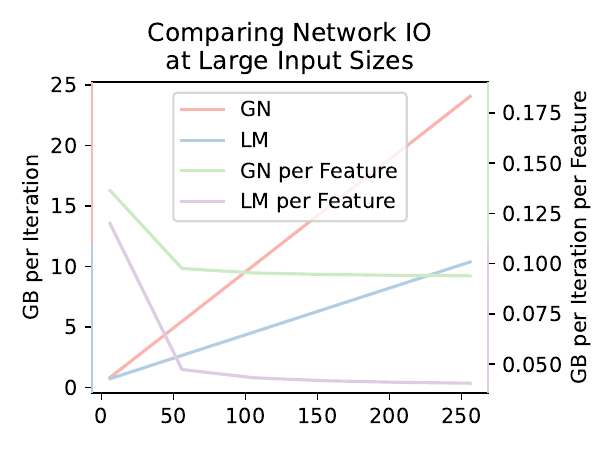}
		\caption{Average number of bytes transmitted between offload servers
			(from garbled circuit generator to evaluator) when computing \algo{}.
			Bytes from evaluator to generator are negligible.
		}
		\label{fig:emp_network_bytes}
	\end{subfigure}
	\caption{Comparison on large input size of Levenberg-Marquardt (LM) and Gauss
		Newton (GN) localization with EMP using \algo{}.
		The heavy cost of large input sizes highlights the importance of input
		preprocessing to reduce the number of input point correspondences.
		Inset axis show such small input sizes between 6 and 12 point correspondences.
		Values are normalized with respect to number of iterations to remove dependence
		on input data and algorithm tunable parameters.
	}
\end{figure*}

\subsection{Performance Evaluation}
Server side performance is evaluated using a desktop with an 11th generation
Intel® core-series CPU.
Each MPC participant is a process on the test machine communicating over
localhost where bandwidth is limited to 2.5 Gbps and latency is introduced
artificially where noted using traffic control.
The rationale for 2.5 Gbps being it is the maximum theoretical bandwidth of
WiFi 6E, the latest standard at the time of writing.
Localization is performed on features from the ETH3D
dataset~\cite{schops2017multi} where each measurement is an average of
randomly selected points where the same random sequence is used across
measurements.
Measurements for which convergence was not achieved are not included in the
results (in these cases the plaintext implementation also did not
converge).
Reported times are wall clock and
implementations are not multithreaded due to their high bandwidth requirements
as later shown.

The performance advantage of \algo{} over the data oblivious baseline are
between two and three orders of magnitude as shown
in~\Cref{fig:loopleak_vs_dataobl}.
It is natural to see why running fewer iterative steps under MPC reduces
runtime, but the magnitude of this difference has implications for the future
of applying general purpose secure computation to localization.
Localization with six points takes 11~seconds to compute rising modestly with
small increases in input size.
While 11~seconds is not within a practical envelope for low latency
localization applications like AR, it is practical for certain robotics
applications, as we later demonstrate.

Next, we consider network communication between {\it client} and offload
servers for the client-server computation offload setting.
EMP, being based on the half-gates garbled circuits
protocol~\cite{zahur2015two}, shares secret inputs via a correlated oblivious
transfer protocol.
Localization is a unique case where neither generator nor evaluator (the two
roles in the protocol the offload servers play) have the secret input data as
they are computing on behalf of a client.
The natural approach to garbled circuits is the generator sends all wire label
pairs to the client, who then forwards the appropriate label from the pair,
based on their secret bits, to the evaluator.
This requires the client receive $2\kappa$ bits per input bit and send $\kappa$
bits per input bit where $\kappa$ is a security parameter, the size of the wire
labels (128~bits by default in EMP's semi-honest protocol).
In the case of six input point correspondences and EMP's default security
parameters, this amounts to the client sending 123~Kb and receiving 246~Kb.
We improve upon this by noticing if the client knows the global circuit delta (a
part of the free XOR optimization~\cite{kolesnikov2008improved}), and circuit
seed, they may generate labels for the evaluator themselves.
From a security perspective this is not an issue, as the client is considered a
trusted third party and is allowed to learn all inputs and outputs of the
computation in the offload setting we consider.
The seed based approach reduces data the client receives to a constant $2\kappa$
bits (circuit delta and seed) while data sent remains $\kappa$ bits per input
bit (wire labels).
In the case of six input point correspondences this amounts to 123~Kb of total
communication instead of 369~Kb (123 + 246) via the natural approach.

\subsection{Limitations}
\label{sec:eval_limitations}
Thus far, performance has been evaluated on small input sizes, between six and
twelve input point correspondences.
While this range is representative of small input sizes for the localization
methods considered and is used by other work~\cite{schneider2018}, the quality
of localization can increase with larger input sizes notably when input data is
noisy.
Feature detection algorithms like SIFT~\cite{lowe2004} typically extract many
more features, around 1000 points per image, but the number which can be
matched to 3D map features is typically lower.
Thus, we consider the performance of input sizes up to 256 points which is
reasonable for popular feature detection algorithms.
Prior localization work considers between 100 and 300 points~\cite{qin2018}.
\Cref{fig:emp_full} focuses on the interaction between the two offload servers,
namely how latency affects runtime and the relationship between
input size and network communication.

Localization even on small input sizes takes around 10~seconds to compute, which
likely precludes time-sensitive applications like AR/VR because the displacement
distance, or error accumulated in the pose as estimated by less accurate means,
grows beyond what is acceptable before the more accurate pose may be computed.
This makes \algo{} better suited for less time sensitive applications like
robotics or offline processing, for example those described in~\Cref{sec:prelim_pploc}.

Comparing the scalability of GN to LM, both algorithms invert a matrix on
every optimization iteration; the key difference is the size of the matrix each
inverts.
GN computes the pseudo-inverse of a $2n \times 6$ Jacobian matrix $\mathbf{J}$
where $n$ is the number of input point correspondences and each row represents
the partial derivative with respect to one degree of freedom in the pose.
LM instead inverts a $6 \times 6$ matrix ($\mathbf{J}^\top \mathbf{J}$ plus an
offset along the diagonal from~\cite{fletcher1971modified}).
Because of the performance impact of inverting large matrices, LM is expected
to perform better than GN on large inputs.
This behavior is not explicitly clear for small numbers of points shown
in~\Cref{fig:loopleak_vs_dataobl}, but becomes important as the input size
scales as shown in~\Cref{fig:emp_full}.
Because LM also outperforms GN on small input sizes, LM is in general the
better approach.
We see each LM iteration of \shortalgo{} takes roughly 20~seconds to compute on
256 input points in~\Cref{fig:emp_full}.
Assuming 8 iterations to converge (a realistic value observed in the ETH3D
dataset), localization completes in 160~seconds.

Given MPC protocols are communication-intensive by nature, we measure the
bytes sent between the two MPC participants in~\Cref{fig:emp_network_bytes}.
Localization using 256 points using GN sends almost 200~GB while LM sends
64~GB, again assuming 8 iterations required for convergence.
In the client-server offload setting, this communication is between the two
offload servers and does not involve the device, however, in the multi-party
setting where input data is distributed between the parties, this high
communication cost is a limiting factor precluding large input sizes.

In \algo{}, we propose executing each iteration of the outer gradient descent
algorithm independently. We argue this approach is secure if invoked a minimum
number of times, a number which depends on a public parameter of the plaintext
algorithm, as discussed in~\Cref{sec:eval_security}. This is
amenable to repeated invocation however it could be seen as a limitation of
in the case where a single or few input images must be localized using
a large stream size. In this case, however, \shortalgo{} is no slower than
the na\"ive data-oblivious adaptation and still benefits from running fewer
SVD iterations driven by the analysis in~\Cref{sec:algo}.

The performance envelope and high network requirements of the proposed methods
are such that localization is best performed using a small number of high
quality feature correspondences.
As such, applying these methods likely requires input
pre-processing to eliminate outlier or even unnecessary inlier point
correspondences.

Before continuing, we briefly mention performance related to trusted execution
environments (TEEs) like Intel® SGX.
Memory requirements to perform localization on the number of points considered
is small enough to fit into secure enclaves, thus we expect performance in TEEs
to be near native (plaintext) speed, making \algo{} orders of magnitude slower
than the 10s of milliseconds required to compute in the clear.
The ``sluggish'' performance of MPC compared to TEE reflects the cost of
eliminating side channels, reducing requirements for specialized hardware, and
eliminating trust required of hardware vendors.

\begin{figure}[t!]
	\centering
	\fbox{\includegraphics[width=\columnwidth]{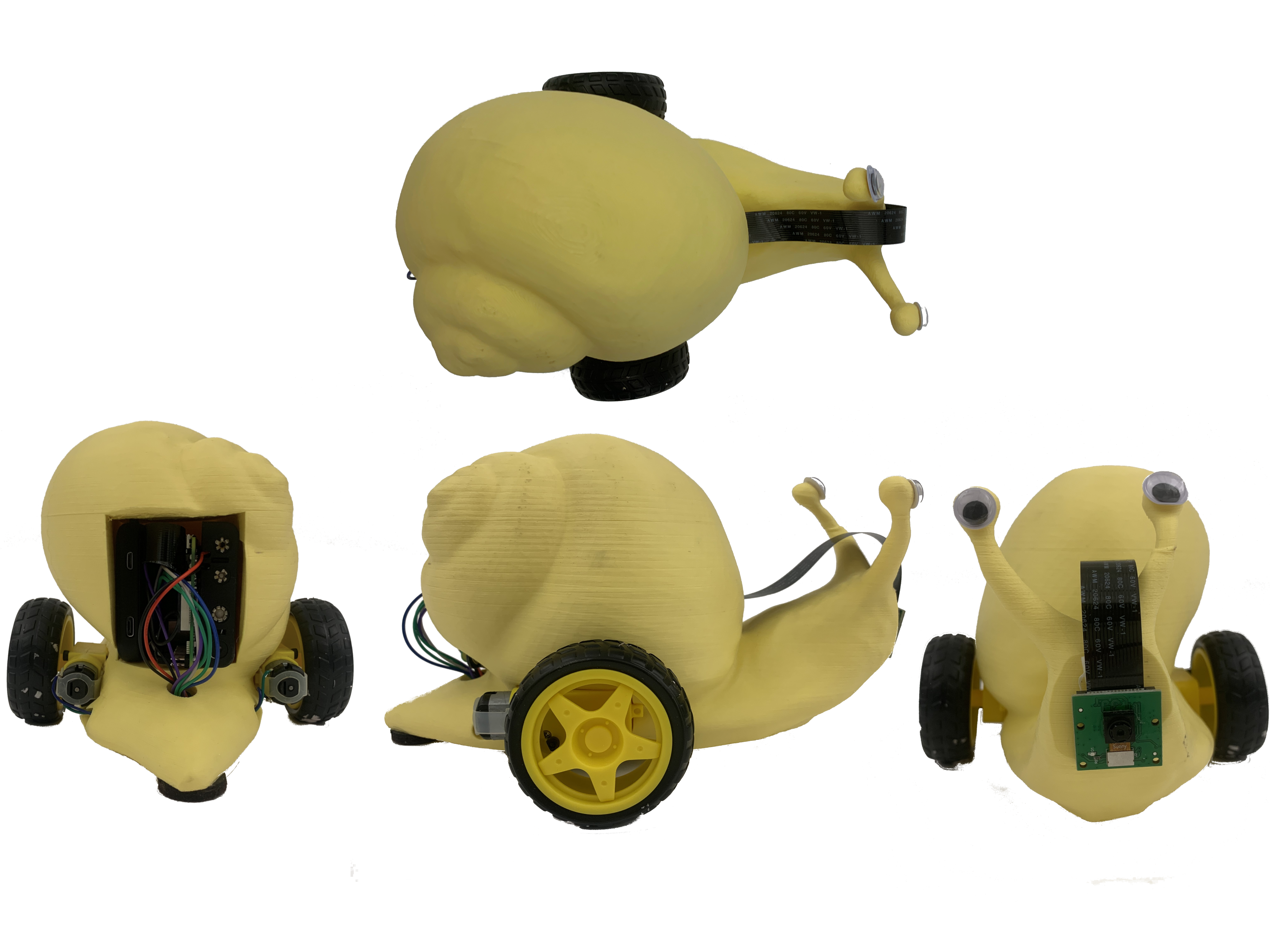}}
	\caption{Robotic snail proof of concept.
		Model adapted from~\cite{snaily}.
	}
	\label{fig:snail}
\end{figure}

\subsection{Proof of Concept: Turbo the Snail}
To demonstrate the feasibility of privacy preserving localization, we built a
Raspberry Pi-based robot acting as a lightweight client which offloads
localization.
The robot is equipped with one RGB camera and uses its onboard WiFi module to
communicate with the offload servers, a maximum bandwidth of 100~Mbps.
The robot is programmed to move to a target location defined within its
environment, leveraging the AprilTag library~\cite{olson2011} to detect marker
images printed on paper in its environment.
The marker corners are detected in the camera image and then matched to a
ground truth 3D map of the marker's position, yielding the 2D and 3D features
used in localization.
Once the pose is known, the robot moves to a predefined target position using
position-based visual servoing~\cite{chaumette2006} via the ViSP
library~\cite{marchand2005}.

In this configuration, the servers learn neither the input feature locations
(2D nor 3D), nor the resultant pose.
By using few high quality input point correspondences, the number of features
is kept low.
The use of specific markers or patterns is commonly used in visual
servoing~\cite{chaumette2006}.
Since privacy preserving localization takes roughly ten seconds to complete on
eight input points as seen in~\Cref{fig:loopleak_vs_dataobl}, we allow the robot
to move 0.1~meters before re-localizing and thus, moves at a rate of
0.01~m/s, hence its form~--~a snail.
We find in practice the snail moves more quickly because the pose estimate from
one localization is used as the initial estimate for the next, reducing the
number of optimization iterations required.
After the initial pose is solved, convergence is frequently achieved in only
one call to \shortalgo{}.

In the context of localization, power consumption is an important metric as
devices are often mobile and battery powered.
To quantify this, the servo motors of the robot are disabled and pose
estimation is performed every 20ms.
In the case of \algo{}, the offload servers return dummy data such that the rate of
localization may be compared to that of plaintext.
Energy consumed by the Raspberry Pi model 3B+ is measured with a TP-Link K115
energy monitor.
Offloading localization reduces power consumption from 3.7 Watts to 2.8 Watts
(24\%) with idle power being 1.7 Watts.
We attribute the increase between idle and \shortalgo{} mostly to polling the
socket while waiting for data to be received from the offload servers,
presenting an opportunity for future optimization.

\vspace{-2ex}
\section{Conclusion}
\label{sec:conclusion}
In this work we introduce a simulation-based definition of security for privacy
preserving localization, and \algo{} (\shortalgo{}) a secure approach to visual
localization.
While \shortalgo{} requires additional rounds of communication, it reduces the
computation and communication required per round by running each localization
iteration independently.
We then demonstrate the practicality of \shortalgo{} with \snailshort{}, the
privacy preserving snail.

\begin{acks}
Thank you anonymous shepherd and reviewers. This work has
been supported by NSF awards CCF-2217070 and CNS-1909769.
\end{acks}

\bibliographystyle{ACM-Reference-Format}

\bibliography{paper}
\end{document}